\documentclass{article}
\usepackage{spconf,amsmath,graphicx}
\usepackage{amsthm,amsfonts}
\usepackage{appendix}
\usepackage{cases}
\usepackage[fontsize=9pt]{scrextend}
\usepackage{color}

\newtheorem{prop}{Proposition}
\newtheorem{theorem}{Theorem}

\title{\large Differential Chaos Shift Keying-based Wireless Power Transfer}
\name{Priyadarshi Mukherjee, Constantinos Psomas, and~Ioannis Krikidis
\thanks{This work has received funding from the European Research Council (ERC) under the European Union's Horizon 2020 research and innovation programme (Grant agreement No. 819819). This work was also co-funded by the European Regional Development Fund and the Republic of Cyprus through the Research and Innovation Foundation, under the projects INFRASTRUCTURES/1216/0017 (IRIDA) and EXCELLENCE/0918/0377 (PRIME).}
}
\address{Department of Electrical and Computer Engineering, University of Cyprus\\
Email: \{mukherjee.priyadarshi, psomas, krikidis\}@ucy.ac.cy}
\begin{document}
%
\maketitle
\begin{abstract}
In this work, we investigate differential chaos shift keying (DCSK), a communication-based waveform, in the context of wireless power transfer (WPT). Particularly, we present a DCSK-based WPT architecture, that employs an analog correlator at the receiver in order to boost the energy harvesting (EH) performance. By taking into account the nonlinearities of the EH process, we derive closed-form analytical expressions for the peak-to-average-power-ratio of the received signal as well as the harvested power. Nontrivial design insights are provided, where it is shown how the parameters of the transmitted waveform affects the EH performance. Furthermore, it is demonstrated that the employment of a correlator at the receiver achieves significant EH gains in DCSK-based WPT systems.
\end{abstract}
\section{Introduction}

According to Ericsson, the wireless traffic is expected to increase more than five times between $2019$ and $2025$ \cite{ericsson}. For applications such as massive machine-type communications, where a large number of devices are deployed, powering or charging becomes critical as well as costly. Hence, low-powered and self-sustainable next generation wireless communication networks is an important and relevant topic of research. In this context, based on the advances made in recent years, wireless power transfer (WPT) can be considered as a suitable candidate, where the devices are wirelessly powered by harvesting energy from ambient/dedicated radio-frequency (RF) signals \cite{harv}. This is achieved by employing a rectifying antenna (rectenna) at the receiver that converts the received RF signals to direct current (DC).

The design of efficient WPT architectures fundamentally relies on accurate mathematical models for the energy harvesting (EH) circuit. The work in \cite{harv} proposes a realistic nonlinear model of the EH circuit, which depends on the circuit characteristics and also enables the design of excitation waveforms that maximize the WPT efficiency. This model has triggered interests in the area of wireless power waveform design, with an objective of maximizing the RF-to-DC conversion efficiency. The authors in \cite{papr} show that the nonlinearity of the rectification process at the EH circuit causes certain waveforms, with high peak-to-average-power-ratio (PAPR) to provide higher output DC power, compared to conventional constant-envelop sinusoidal signals. Based on this observation, there are some works, which investigate the effect of the transmitted symbols and modulation techniques on WPT. By considering the nonlinear EH model proposed in \cite{harv}, the authors in \cite{wdesg} investigate the use of multisine waveforms for WPT due to their high PAPR. The work in \cite{hparam} proposes a simultaneous wireless information and power transfer (SWIPT) architecture based on the superposition of multi-carrier unmodulated and modulated waveforms at the transmitter. Apart from the multisine waveforms, experimental studies demonstrate that due to their high PAPR, chaotic waveforms outperform conventional single-tone signals in terms of WPT efficiency \cite{chaosexp2}.

Due to its properties such as sensitivity to initial data and aperiodicity, chaotic waveforms have been extensively used in the past to improve the performance of wireless communication systems. In this context, the non-coherent modulation technique of differential chaos shift keying (DCSK) is one of the most widely studied chaotic signal-based communication system \cite{ch2}. The majority of the related works focus on the error performance of such systems for various scenarios. To exploit the benefits of both DCSK and WPT, there are few works in the literature that investigate SWIPT in a chaotic framework, e.g. \cite{chaoswipt1,chaoswipt3,chaoswipt4}. In \cite{chaoswipt1}, a short-reference DCSK-based SWIPT architecture is proposed to achieve higher data rate than the conventional system. A chaotic multi-carrier system is investigated in a SWIPT framework via the sub-carrier index to reduce the energy consumption \cite{chaoswipt3}. In \cite{chaoswipt4}, adaptive link selection for buffer-aided relaying is investigated in a DCSK-SWIPT architecture, where two link-selection schemes based on harvested energy are proposed.

However, the above studies consider a simplified linear model for the EH, and as a result, they are independent of the circuit characteristics as well as the design of excitation waveforms  \cite{harv}. Motivated by this, in this paper, we present a DCSK-based WPT architecture by taking into account the nonlinearities of the EH process. Specifically, the contribution of this work is three fold. Firstly, we propose a novel WPT architecture, where an analog correlator is employed at the receiver in order to boost the EH performance. Secondly, the flexibility of the proposed architecture is demonstrated, as the correlator can control the PAPR of the received signal at the harvester. Finally, analytical expressions of PAPR and the harvested DC are derived for both cases, i.e. with and without the correlator. The analytical framework provides a convenient methodology for obtaining nontrivial insights into how key parameters affect the performance.

\section{A Chaotic Signal-based WPT System Architecture} \label{SM}

\subsection{System model}
\begin{figure}[!t]
\centering\includegraphics[width=\linewidth]{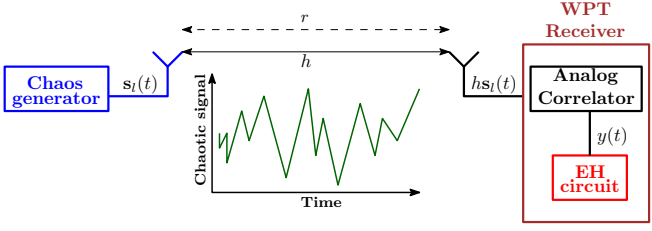}
\caption{Proposed architecture for DCSK-based WPT.}
\label{fig:model}
\vspace{-4mm}
\end{figure}

We consider a point-to-point WPT set-up, where the transmitter employs a DCSK generator and the receiver consists of an analog correlator, followed by an EH circuit, as shown in Fig. \ref{fig:model}. Note that, DCSK signals, until now, have been mainly considered for information transfer. However, here we focus on the WPT aspect and ignore the information side of the signal. Furthermore, the decoding of these signals can be found in \cite{ch2,chaoswipt1,chaoswipt4}. We assume that the received power is proportional to $r^{-\alpha}$, where $r$ is the transmitter-receiver (Tx-Rx) distance and $\alpha$ denotes the path-loss exponent. Moreover, a Rayleigh distributed flat fading channel $h$ is considered with unit mean power, i.e. $f_{|h|}(\gamma)=2\gamma e^{-\gamma^2},$ $\forall$ $\gamma \geq 0.$

By considering a circuit-based nonlinear model of the harvester circuit, the output DC current is approximated in terms of the input signal $y(t)$ as \cite{harv}
\vspace{-2mm}
\begin{equation} \label{brunoeh}
z_{\text{DC}}=k_2R_{ant}\mathbb{E} \{ |y(t)|^2 \}+k_4R_{ant}^2\mathbb{E} \{ |y(t)|^4 \},
\end{equation}
which is a monotonically increasing function of the DC component of the current at the rectifier output and the parameters $k_2,k_4,$ and $R_{ant}$ are constants determined by the characteristics of the circuit.

\subsection{Chaotic signals}

Each DCSK symbol is represented by two consecutive chaotic signal components; the first one serves as the reference, while the second carries the data $\pm1$. The $l$-th transmitted DCSK symbol is characterized by a sequence of $2\beta$, $\beta>0$, samples of the chaotic basis signal, in which the $k$-th bit is given by \cite{ch2}
\vspace{-2mm}
\begin{align}  \label{sym}
s_{l,k}=\begin{cases} 
x_k, & k=2(l-1)\beta+1,\dots,(2l-1)\beta,\\
d_lx_{k-\beta}, & k=(2l-1)\beta+1,\dots,2l\beta,
\end{cases}&
\end{align}
where $d_l=\pm1$ is the $l$-th information bit, and $x_k$ is the chaotic basis signal. As $2\beta$ chaotic samples are being used to spread each information bit, $\beta$ is defined as the \textit{spreading factor}. Thus, the overall $l$-th transmitted symbol at time $t$ is ${\bf{s}}_l(t)=[s_{l,1}(t),s_{l,2}(t),\cdots,s_{l,2\beta}(t)]$. Due to its good correlation properties, we consider the Chebyshev chaotic map of degree $\xi$, where the invariant probability density function (PDF) of $x_k$, namely $f_X(x)$ is \cite{ch2}
\vspace{-2mm}
\begin{align}  \label{spdf}
f_X(x)=\begin{cases} 
\frac{1}{\pi\sqrt{1-x^2}}, & |x|< 1,\\
0, & \text{otherwise}.
\end{cases}&
\end{align}

\subsection{Analog correlator}   \label{AC}

The proposed WPT architecture employs an analog correlator, followed by an EH rectifier circuit. An analog correlator essentially consists of a series of $(\psi-1)$ delay blocks, where $\psi$ is a positive integer; the rationale behind this application is that, the signal can be effectively integrated over a certain time interval \cite{anaco2}. Thus, an ideal $\psi$-bit analog correlator provides an output signal $y(t)=h\int_{\nu=-\psi T}^0 {\bf{s}}_l(t-\nu)d\nu$
where $h$ is the channel coefficient, $T$ is the bit period, and ${\bf{s}}_l(t)$ is the chaotic input signal \cite{anaco2}. In what follows, for the sake of simplicity, we will consider $\psi$ equal to the transmitted DCSK symbol length, i.e. $\psi=2\beta$. Hence, the correlator output $y_l(t)$ for the $l$-th transmitted symbol ${\bf{s}}_l(t)$ is
\vspace{-2mm}
\begin{equation}    \label{corr}
y_l(t)=\sqrt{P_t}h\sum\limits_{k=1}^{2\beta} s_{l,k}(t),
\end{equation}
where $P_t$ is the transmission power. Note that $\psi=1$ corresponds to the conventional case without a correlator. We state the following proposition, that refers to the effect of the analog correlator on the signal's PAPR.

\begin{prop}    \label{prp1}
The signal PAPR at the harvester input is
\end{prop}
\vspace{-8mm}
\begin{align}  \label{prop1}
\mathrm{PAPR}=\begin{cases} 
2, & \text{without correlator} \:\: (\psi=1),\\
4\beta, & \text{with correlator} \:\: (\psi=2\beta).
\end{cases}&
\end{align}
\begin{proof}
See Appendix A.
\end{proof}

\noindent From Proposition \ref{prp1}, we observe that the correlator can control the value of PAPR through the design parameter $\psi$. Also, since high PAPR signals are desirable for WPT \cite{papr}, the correlator can significantly enhance the EH performance of the DCSK signals.

\section{Chaotic Signal-based \\ Wireless Power Transfer} \label{harvsec}

In this section, we investigate the effect of the DCSK waveform on the WPT performance. Specifically, we evaluate $z_{\text{DC}}$ for both cases with/without the correlator. For the sake of simplicity, we will use $\rho_1 = k_2R_{ant}P_t$ and $\rho_2=k_4R_{ant}^2P_t^2$. Hence, from \eqref{brunoeh} and (\ref{corr}), the harvested DC with a correlator is
\vspace{-2mm}
\begin{equation}
z_{\text{C}}=\!r^{-\alpha}\rho_1\mathbb{E}\left\lbrace\! \left(  |h|\sum\limits_{k=1}^{2\beta}s_k \right)^{\!2} \right\rbrace + r^{-2\alpha}\rho_2\mathbb{E}\left\lbrace\! \left(  |h|\sum\limits_{k=1}^{2\beta}s_k \right)^{\!4} \right\rbrace, \quad\label{zdef1}
\end{equation}
\vspace{-2mm}
and for the case without a correlator is
\vspace{-1.2mm}
\begin{equation}
z_{\text{NC}}=\!r^{-\alpha}\rho_1\mathbb{E}\left\lbrace |h|^2\sum\limits_{k=1}^{2\beta}s_k^2 \right\rbrace  + r^{-2\alpha}\rho_2\mathbb{E}\left\lbrace |h|^4\sum\limits_{k=1}^{2\beta}s_k^4 \right\rbrace, \label{zdef2}
\end{equation}
where the expectation is taken over $h$ and $s_k$.

We now state the following two theorems, which provide closed-form expressions for $z_{\text{C}}$ and $z_{\text{NC}}$.

\begin{theorem} \label{nakath1}
When a correlator is employed, the harvested DC is
\vspace{-2mm}
\begin{align}  \label{theocnaka}
z_{\text{C}}=\begin{cases} 
r^{-\alpha}\rho_1+6r^{-2\alpha}\rho_2, & \beta=1,\\
r^{-\alpha}\rho_1\beta+12r^{-2\alpha}\rho_2\beta^2, & \beta>1.
\end{cases}&
\end{align}
\end{theorem}

\begin{proof}
See Appendix B.
\end{proof}

\noindent Next, we consider the case of a WPT receiver without the analog correlator, which is given by the following theorem.
\begin{theorem} \label{nakath2}
When a correlator is not employed, the harvested DC is given by
\vspace{-2mm}
\begin{equation}  \label{theonaka}
z_{\text{NC}}=r^{-\alpha}\rho_1\beta+\frac{3}{2}r^{-2\alpha}\rho_2\beta.
\end{equation}
\end{theorem}
\begin{proof}
See Appendix C.
\end{proof}

From Theorems \ref{nakath1} and \ref{nakath2}, we observe that $z_{\text{C}}$ and $z_{\text{NC}}$ is a quadratic and linear function of $\beta$, respectively. Therefore, for fixed $r$, $\rho_1$ and $\rho_2$, we have $z_{\text{C}}>z_{\text{NC}}$, $\forall$ $\beta$. This clearly demonstrates the impact of the correlator on the WPT performance. Also, note that if a linear EH model is used, i.e. only the second order term of \eqref{brunoeh} is considered, both are equivalent.

Now, with a slight abuse of notation, let $r_{\rm C}$ and $r_{\rm NC}$ be the Tx-Rx distances for $z_{\rm C}$ and $z_{\rm NC}$, respectively. We can show that $z_{\rm C} \geq z_{\rm NC}$, even when $r_{\rm C} > r_{\rm NC}$. To achieve this, $\beta$ needs to satisfy 
\vspace{-2mm}
\begin{equation} \label{blimit}
\beta > \frac{\rho_1 \left(r_{\text{NC}}^{-\alpha}-r_{\text{C}}^{-\alpha}  \right) +1.5\rho_2 r_{\text{NC}}^{-2\alpha} }{12 \rho_2 r_{\text{C}}^{-2\alpha} },
\end{equation}
\vspace{-0.8mm}
which can be easily derived from \eqref{theocnaka} and \eqref{theonaka}. 

Next, we compare the proposed architecture with the existing multisine waveform based EH framework \cite{wdesg}. We know from \cite{wdesg}, that the achieved EH performance when an $N$-tone multisine waveform is used, results in a $z_{\text{DC}}$ whose linear term is independent of $N$ and the nonlinear term is linearly dependent on $N$. Furthermore, for a given bandwidth, the number of tones in a multisine waveform cannot be decided arbitrarily. A large $N$ implies very small inter-tone spacing, which results in low output DC voltage at the harvester. On the other hand, a too small $N$ suggests infinitely large inter-tone spacing, i.e. most of the signal gets filtered out by the low pass filter at the  harvester. Hence, multisine waveforms can enhance the WPT performance only if an optimum inter-tone spacing is selected, and experimentally it is observed that we have $N<10$ \cite{bweff}. In the case of DCSK-based waveforms, as observed from Theorem $1$, the linear and nonlinear terms of $z_{\text{C}}$ are proportional to $\beta$ and $\beta^2$, respectively; also, we usually have $\beta \gg 1$ \cite{chaoswipt1}. Hence, the WPT performance of the proposed DCSK-based WPT architecture is significantly greater than the multisine waveform-based EH.

\section{Numerical Results}

We consider a transmission power $P_t=30$ dBm and path-loss exponent $\alpha=4$. The parameters considered for the WPT model are $k_2=0.0034,k_4=0.3829,$ and $R_{ant}=50$ $\Omega$ \cite{wdesg}. Recall that, $\psi=2\beta$ and $\psi=1$ correspond to a WPT receiver with and without the analog correlator, respectively.

Fig. \ref{fig:comp} demonstrates the performance of the proposed DCSK-based WPT architecture in terms of harvested DC, where we observe a significant improvement. Although the analytical expressions for $z_{\text{DC}}$, when $\beta>1$, are central limit theorem (CLT)-based, we observe that the theoretical results match very closely with the simulation results. The significant gain in WPT performance with the correlator is related to the high PAPR, which is a function of the spreading factor $\beta$, as stated in Proposition \ref{prp1}. Moreover, we note that, with/without the correlator at the receiver, the harvested DC with Tx-Rx distance $30$ m is less compared to the harvested DC with Tx-Rx distance $20$ m; this is intuitive due to the path-loss factor. Finally, we also observe that $z_{\text{C}}$ with Tx-Rx distance of $30$ m outperforms $z_{\text{NC}}$ with a smaller Tx-Rx distance of $20$ m when we have $\beta>52$. This observation matches the lower bound proposed in \eqref{blimit}, which results in $\beta>52$ for the considered set of system parameters.

\begin{figure}[!t]
\centering\includegraphics[width=0.88\linewidth]{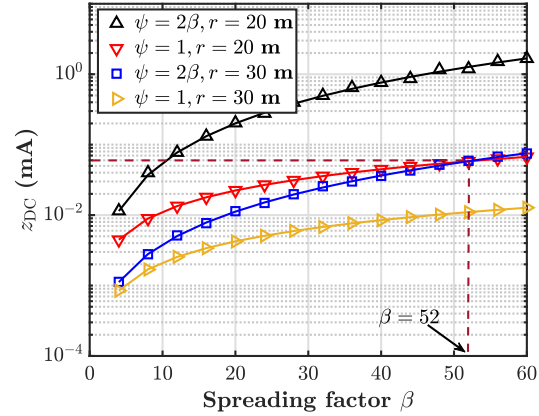}
\vspace{-4mm}
\caption{\small Effect of spreading factor on $z_{\text{DC}}$; lines correspond to analysis and markers correspond to simulation results.}
\label{fig:comp}
\vspace{-2mm}
\end{figure}

\section{Conclusion}

In this work, we investigate DCSK-based WPT, by taking into account the nonlinearities of the EH process. Specifically, we use a communication-based waveform for the purpose of WPT and also propose a new analog correlator-aided WPT receiver. By considering a Rayleigh fading scenario, we derive analytical expressions of the harvested DC for both cases, i.e. with and without the correlator. Our results show the significant gains that can be achieved by our proposed approach. A promising extension of this work is to investigate this design framework in the context of SWIPT.

\section{Appendix A: Proof of Proposition 1}

Without a correlator at the receiver, i.e. for $\psi=1$, the PAPR corresponding to $l$-th transmitted symbol is
\vspace{-2mm}
\begin{equation}
\mathrm{PAPR}=\frac{\max\limits_{l} \left\lbrace  \sum\limits_{k=2(l-1)\beta+1}^{2l\beta}|h_l|^2s_{l,k}^2 \right\rbrace }{\mathbb{E}\left\lbrace \sum\limits_{k=2(l-1)\beta+1}^{2l\beta}|h_l|^2s_{l,k}^2 \right\rbrace }.
\end{equation}
By considering a channel instance $h_l$ and the PDF of $x_k$ from \eqref{spdf}, 
we obtain $\max\limits_{l} \left\lbrace  \sum\limits_{k=2(l-1)\beta+1}^{2l\beta}|h_l|^2s_{l,k}^2 \right\rbrace=2|h_l|^2\beta$ and
\vspace{-2mm}
\begin{equation} \label{deno}
\mathbb{E}\left\lbrace \sum\limits_{k=2(l-1)\beta+1}^{2l\beta}|h_l|^2s_{l,k}^2 \right\rbrace=|h_l|^2\!\!\!\!\sum\limits_{k=2(l-1)\beta+1}^{2l\beta} \mathbb{E}\{s_{l,k}^2\} =|h_l|^2\beta.
\end{equation}
Hence, for $\psi=1$, we have $\rm PAPR=\frac{2|h_l|^2\beta}{|h_l|^2\beta}=2$. On the other hand, for $\psi=2\beta$, i.e. with a correlator, we have
\vspace{-2mm}
\begin{equation}
\mathrm{PAPR}=\frac{\max\limits_{l}\left\lbrace \left( \sum\limits_{k=2(l-1)\beta+1}^{2l\beta}|h_l|s_{l,k} \right)^2 \right\rbrace}{\mathbb{E} \left\lbrace \left( \sum\limits_{k=2(l-1)\beta+1}^{2l\beta}|h_l|s_{l,k} \right)^2 \right\rbrace },
\end{equation}
where from (\ref{spdf}), we have $\max\limits_{l}\left\lbrace \sum\limits_{k=2(l-1)\beta+1}^{2l\beta}|h_l|s_{l,k} \right\rbrace^2=4|h_l|^2\beta^2$. Note that for chaotic sequences generated by the Chebyshev map, we have $\mathbb{E}[s_{l,i}s_{l,j}]=0$ for $i \neq j$ \cite[Eq.~55]{nrml}. As such, $\mathbb{E} \left\lbrace \left( \sum\limits_{k=2(l-1)\beta+1}^{2l\beta}|h_l|s_{l,k} \right)^2 \right\rbrace=|h_l|^2\beta,$
which follows from (\ref{deno}). Hence, we have $\mathrm{PAPR}=\frac{4|h_l|^2\beta^2}{|h_l|^2\beta}=4\beta$.

\section{Appendix B: Proof of Theorem 1}

The quantity $|h|\sum\limits_{k=1}^{2\beta}s_k$ in \eqref{zdef1} can be alternatively written as
\vspace{-2mm}
\begin{equation} \label{nchrv}
s=|h|\sum\limits_{k=1}^{2\beta}s_k=|h|\sum\limits_{k=1}^{2\beta}x_k=|h|(1+d)\sum\limits_{k=1}^{\beta}x_k.
\end{equation}
\vspace{-0.6mm}
By assuming equally likely transmission of $d=\pm 1$ and if we consider $a=1+d$, it results in the PDF $f_A(a)=\frac{1}{2},$ $a\in \{0,2\}$. Hence, $s$ can be written as a product of random variables, i.e. $S=|h|AV$, where we have $V=\sum\limits_{k=1}^{\beta}X_k$, i.e. the PDF $f_V(v)$ is the $(\beta-1)$-fold convolution of $f_X(x)$ \cite{papoulis}. As a closed-form expression cannot be obtained for $f_V(v)$ with any arbitrary $\beta$, we obtain an analytical expression of $f_S(s)$ for $\beta=1$ and a well-approximated CLT-based expression of $f_S(s)$ for $\beta > 1$, respectively.

\subsection{Case $\beta=1$}

For $\beta=1,$ we have $s=|h|(1+d)x_1$. Therefore the cumulative distribution function of $s$ conditioned on $h$ is
\vspace{-2mm}
\begin{align}
F_{S|H}(s|h)&=\sum\limits_{i=0,2}\text{Pr}(AX<s|A=i,h)\text{Pr}(A=i) \nonumber \\
&=\frac{1}{2}\left[ \frac{1}{\pi}\sin^{-1}\left( \frac{s}{2|h|} \right) +\frac{1}{2}+ \mathbf{1}_{[0,\infty)}\left( \frac{s}{|h|}\right) \right].
\end{align}
As a result, the pdf of $s$ conditioned on $h$ is obtained as $f_S(s|h)=\frac{\partial F_{S|H}(s|h)}{\partial s}=\frac{1}{2\pi\sqrt{4|h|^2-s^2}}$. Accordingly we obtain $f_S(s)$ by unconditioning on $h$, i.e. $f_S(s)=\int\limits_{\frac{s}{2}}^{\infty} \frac{\alpha e^{-\alpha^2}}{\pi\sqrt{4\alpha^2-s^2}} d\alpha =\frac{e^{-\frac{s^2}{4}}}{4\sqrt{\pi}}.$ Based on \eqref{zdef1}, we need the PDFs corresponding to $S^2=Z$ and $S^4=P$. The CDF of $Z$ is obtained as $F_Z(z)=\mathbb{P}(-\sqrt{z} \leq S \leq \sqrt{z})=F_S(\sqrt{z})-F_S(-\sqrt{z})$, which results in the PDF $f_Z(z)=\frac{\partial F_Z(z)}{\partial z}=\frac{1}{4\sqrt{\pi z}}e^{-\frac{z}{4}}$. As $Z$ is non-negative in nature, the CDF of $P$ is  $F_P(p)=\mathbb{P}(Z^2 \leq p)=\mathbb{P}(Z \leq \sqrt{p})=F_Z(\sqrt{p})$, i.e.
\vspace{-1.6mm}
\begin{equation}    \label{ppdfapp2}
f_P(p)=\frac{\partial F_P(p)}{\partial p}=\frac{1}{2\sqrt{p}} f_Z(\sqrt{p})=\frac{1}{8\sqrt{\pi}p^{3/4}}e^{-\frac{\sqrt{p}}{4}}.
\end{equation}
\vspace{-1.6mm}
Hence,  we obtain $z_{\text{C}}$ as
\begin{equation}
z_{\text{C}}=r^{-\alpha}\rho_1\mathbb{E}[Z] + r^{-2\alpha}\rho_2\mathbb{E}[P] =r^{-\alpha}\rho_1+6r^{-2\alpha}\rho_2.
\end{equation}

\subsection{Case $\beta >1$}

We first perform a CLT-based characterization of $\left\lbrace |h|\sum\limits_{k=1}^{2\beta}x_k \right\rbrace^2$ and $\left\lbrace |h|\sum\limits_{k=1}^{2\beta}x_k \right\rbrace^4,$ followed by deriving an approximate closed-form expression of $z_{\text{C}}$. We know that $s=|h|(1+d)\sum\limits_{k=1}^{\beta}x_k,$ where for a sufficiently large $\beta$, CLT states that $v=\sum\limits_{k=1}^{\beta}x_k$ will follow a Gaussian distribution with mean $\mu=\beta\mathbb{E}[X]$ and variance $\sigma^2=\beta \text{Var}[X]$, where $\mathbb{E}[X]=0$ and $\text{Var}[X]=\frac{1}{2}.$ Hence, we have $V \sim \mathcal{N}\left( 0,\frac{\beta}{2} \right)$ and by following a similar methodology as in the $\beta=1$ case, we derive the CDF of $S$
\vspace{-2mm}
\begin{equation} \label{bgf1}
F_{S|H}(s|h)=\frac{1}{2}\left[ \frac{1}{2}\left(1+\text{erf}\left[\frac{s}{2|h|\sqrt{\beta}} \right]  \right)+ \mathbf{1}_{[0,\infty)}\left( \frac{s}{|h|}\right)  \right].
\end{equation}
\vspace{-2mm}
By differentiating (\ref{bgf1}), we obtain $f_{S|H}(s|h)=\frac{1}{4|h|\sqrt{\pi\beta}}e^{-\frac{s^2}{4|h|^2\beta}},$ which by unconditioning on $h$ yields $f_S(s)=\frac{1}{4\sqrt{\beta}}e^{-\sqrt{\frac{s^2}{\beta}}}$. Now we obtain the pdf of $s^2(=z)$ and $s^4(=p)$ as $f_Z(z)=\frac{1}{4\sqrt{\beta z}}e^{-\sqrt{\frac{z}{\beta}}}$ and $f_P(p)=\frac{1}{8\sqrt{\beta}p^{\frac{3}{4}}}e^{-\sqrt{\frac{\sqrt{p}}{\beta}}}$ respectively. Hence
\vspace{-2mm}
\begin{equation}
z_{\text{C}}=r^{-\alpha}\rho_1\mathbb{E}[Z] + r^{-2\alpha}\rho_2\mathbb{E}[P] =r^{-\alpha}\rho_1\beta+12r^{-2\alpha}\rho_2\beta^2.
\end{equation}
\vspace{-0.8mm}
It is interesting to observe that although we have obtained a CLT-based analytical expression, the approximation error is always less than $5\%$ irrespective of the value of $\beta.$

\section{Appendix C: Proof of Theorem 2}

The proof follows steps similar with the analysis in Theorem \ref{nakath1}, i.e. separately considering the cases of $\beta=1$ and $\beta>1$.

\subsection{Case $\beta=1$}

Let $Y=X^2$, where $X$ follows the PDF as stated in (\ref{spdf}). Hence, we obtain $F_Y(y)=\mathbb{P}(X^2 \leq y)=\mathbb{P}(-\sqrt{y} \leq X \leq \sqrt{y})= F_X(\sqrt{y})-F_X(-\sqrt{y})$ and so the PDF is given by
\vspace{-2mm}
\begin{equation} \label{theo1n}
f_Y(y)=\frac{1}{2\sqrt{y}}\left[f_X(\sqrt{y})+f_X(-\sqrt{y}) \right]=\frac{1}{\pi\sqrt{y(1-y)}}.
\end{equation}
Thus, the PDF of $\Delta=\sum\limits_{k=1}^{2\beta}x_k^2$ corresponding to $\beta=1$ is equal to $f_{\Delta}(\delta)=\frac{1}{2}f_Y\left(\frac{\delta}{2} \right)=\frac{1}{\pi \sqrt{\delta(2-\delta)}}$. We obtain the PDF of $|h|^2\Delta$ as $f_{\Delta|H}(\delta|h)=\frac{1}{|h|^2}f_{\Delta}\left(\frac{\delta}{|h|^2} \right) =\frac{1}{\pi \sqrt{\delta(2|h|^2-\delta)}}$, which results in $f_{\Delta}(\delta)=\int\limits_{0}^{\infty}\frac{e^{-\alpha}}{\pi \sqrt{\delta(2\alpha-\delta)}}d\alpha=\frac{1}{\sqrt{2\pi \delta}}e^{-\frac{\delta}{2}}.$

Next, we require the PDF of $|h|^4\Theta$, where $\Theta=\sum\limits_{k=1}^{2\beta}x_k^4$ and from (\ref{theo1n}), the PDF of $Y=X^2$ is $f_Y(y)=\frac{1}{\pi\sqrt{y(1-y)}}$. Hence, we obtain the CDF of $G=Y^2$ as $F_G(g)=\mathbb{P}(Y^2 \leq g)=\mathbb{P}(0\leq Y \leq \sqrt{g})=F_Y(\sqrt{g})$, i.e. the PDF $f_G(g)$ is
\vspace{-2mm}
\begin{equation}  \label{dpdfn}
f_G(g)=\frac{\partial F_G(g)}{\partial g}=\frac{1}{2\sqrt{g}}f_Y(\sqrt{g})=\frac{1}{2\pi\sqrt{g^{3/2}(1-\sqrt{g})}}.
\end{equation}
Therefore, the PDF of $\Theta$ is obtained by standard transformation of random variables as $f_{\Theta}(\theta)=\frac{1}{2^{5/4}\pi\sqrt{1-\sqrt{\frac{\theta}{2}}}\theta^{3/4}}$, which results in the PDF of $|h|^4\Theta=|h|^4\sum\limits_{k=1}^{2\beta}x_k^4$ as $f_{\Theta|H}(\theta|h)=\frac{1}{|h|^4}f_{\Theta}\left(\frac{\theta}{|h|^4} \right) =\frac{1}{2\pi\sqrt{\theta^{3/2}(\sqrt{2}|h|^2-\sqrt{\theta})}}.$ By unconditioning on $|h|$, we obtain $f_{\Theta}(\theta)=\frac{1}{2^{\frac{5}{4}}\sqrt{\pi}\theta^{\frac{3}{4}}}e^{-\sqrt{\frac{\theta}{2}}}$. Hence, we obtain
\vspace{-2mm}
\begin{equation}
z_{\text{NC}}=r^{-\alpha}\rho_1\mathbb{E}[\Delta]+r^{-2\alpha}\rho_2\mathbb{E}[\Theta]=r^{-\alpha}\rho_1+\frac{3}{2}r^{-2\alpha}\rho_2.
\end{equation}

\subsection{Case $\beta>1$}

As the proof can be obtained by following a similar procedure as described in Appendix B, i.e. we follow a CLT-based approach, the detailed derivation has been omitted due to space limitation. Finally, by combining the two cases of $\beta=1$ and $\beta>1$,we obtain $z_{\text{NC}}=r^{-\alpha}\rho_1\beta+\frac{3}{2}r^{-2\alpha}\rho_2\beta,$ $ \forall$ $\beta.$

\vfill\pagebreak
\bibliographystyle{IEEEbib}
\bibliography{Draft}
\end{document}